\documentclass[conference]{IEEEtran}
\IEEEoverridecommandlockouts

\usepackage[numbers]{natbib}
\usepackage[margin={0.80in,0.79in}]{geometry}

\usepackage{graphicx}
\usepackage{caption}
\usepackage{subcaption}
\usepackage{etoolbox}
\usepackage{scalerel}
\usepackage{hyperref}
\newcommand*{\img}[1]{%
    \raisebox{-.2\baselineskip}{%
        \includegraphics[
        height=\baselineskip,
        width=\baselineskip,
        keepaspectratio,
        ]{#1}%
    }%
}
\usepackage{siunitx}
\usepackage{pifont}
\usepackage{graphicx}
\usepackage{amsmath,amssymb,wasysym,stix}
\usepackage{textcomp}
\usepackage{algorithm}
\usepackage{algpseudocode}

\usepackage{cancel}
\usepackage{booktabs}
\usepackage{url}
\usepackage{paralist}
\usepackage{color}
\usepackage[usenames, dvipsnames]{xcolor}
\definecolor{FGreen}{RGB}{6,253,105}
\usepackage{epsfig} % for postscript graphics files
\usepackage{bm} % assumes amsmath package installed
\usepackage{listings,lstautogobble}

\usepackage{xspace}
\usepackage{soul,xcolor}
\usepackage{verbatim}
\usepackage{epstopdf}
\usepackage{listings}
\usepackage{parcolumns}
\usepackage{color}
\usepackage{mathtools}

\usepackage{xspace}
\usepackage{soul,xcolor}
\usepackage{float}
\usepackage{verbatim}
\usepackage[utf8]{inputenc}

\usepackage{rotating,multirow}
\usepackage{authblk}

\usepackage[most]{tcolorbox}

\usepackage{amsthm}
\newtheorem{definition}{Definition}

\newtheorem{theorem}{Theorem}
\newtheorem{lemma}{Lemma}

  \renewcommand{\footnotesize}{\fontsize{8pt}{11pt}\selectfont}

\DeclareMathAlphabet{\mathcal}{OMS}{cmsy}{m}{n}

\DeclarePairedDelimiter{\norm}{\lVert}{\rVert}
\DeclareMathOperator*{\argminB}{argmin}   % Jan Hlavacek

\newcommand{\zono}[1]{\langle #1 \rangle}

\newcommand{\trace}{\text{Tr}}

\begin{document}

\title{\vspace{.85in}Differentially Private Set-Based Estimation\\ Using Zonotopes}

\author{Mohammed M. Dawoud$^{1}$, Changxin Liu$^{2}$, Amr Alanwar$^{1}$, and Karl H. Johansson$^{2}$
\thanks{This paper has received funding from the European Union's Horizon 2020 research and innovation program under grant agreement No. 830927.}
\thanks{$^{1}$School of Computer Science and Engineering, Constructor University, Bremen, Germany (Email: \{mdawoud, aalanwar\}@constructor.university)}
\thanks{$^{2}$Division of Decision and Control Systems, School of Electrical Engineering and Computer Science, KTH Royal Institute of Technology, and Digital Futures, Stockholm, Sweden  (Email: \{changxin, kallej\}@kth.se)}
}

\vspace{1in}
\maketitle
\begin{abstract}

For large-scale cyber-physical systems, the collaboration of spatially distributed sensors is often needed to perform the state estimation process. Privacy concerns naturally arise from disclosing sensitive measurement signals to a cloud estimator that predicts the system state. To solve this issue, we propose a differentially private set-based estimation protocol that preserves the privacy of the measurement signals. 
Compared to existing research, our approach achieves less
privacy loss and utility loss using a numerically optimized truncated noise distribution.
The proposed estimator is perturbed by weaker noise than the analytical approaches in the literature to guarantee the same level of privacy, therefore improving the estimation utility. Numerical and comparison experiments with truncated Laplace noise are presented to support our approach.
Zonotopes, a less conservative form of set representation, are used to represent estimation sets, giving set operations a computational advantage. The privacy-preserving noise anonymizes the centers of these estimated zonotopes, concealing the precise positions of the estimated zonotopes. 

\end{abstract}
% \begin{IEEEkeywords}


\section{Introduction} \label{sec:intro}
The widespread use of automated and intelligent systems and the Internet of Things (IoT) in our daily lives has become an apparent fact. One can see it in smart homes, smart devices, Intelligent Transportation systems (ITS), and intelligent weather forecasting systems. These large-scale systems often need their participating entities to share some data, e.g., vehicle location or pedestrian destination, medical measurements, financial information, and users' habits \cite{8686209}. Sharing these data within the system can cause privacy vulnerabilities, e.g., individuals' positions may be inferred from location data shared with routing applications such as Waze and Google Maps \cite{10.1145/2508859.2516735,6256763,proloc}. 
Another example of a privacy vulnerability is identity disclosure for a driver when he shares his real identity in vehicular ad-hoc networks (VANETs) of ITS. Thus, there is a pressing need for privacy-preserving mechanisms for data publishing within such systems. Motivated by this, this work is dedicated to privacy-preserving state estimation.

\subsection{Set-Based Estimation}

Estimation is a ubiquitous phenomenon that plays a crucial role in decision-making across various domains.
From an estimated value-based perspective, it can be categorized into point-wise and set-based estimation. The former uses measurements to calculate a point-valued estimate, such as the location of an autonomous vehicle and the distance between space objects \cite{proloc}. Alternatively, set-based estimation uses the measurements and the system model, if available, to predict a set that encloses the system's state. Supported by the development of sensors' capabilities, the widespread of autonomous systems, which massively deploy set-based estimation, can be found in applications including robots and autonomous vehicles \cite{doi:10.1177/1729881419839596,s21062085,vandana}.
 
 There are three approaches to set-based estimation: the model-based approach, the data-driven approach, and the hybrid approach. The model-based approach depends on the system model and the sensors' measurements to estimate the system's states \cite{s21062085}. Some of the filters used in the model-based approach can estimate states for systems with a linear model and a Gaussian noise distribution, such as the Kalman filter \cite{7549014}. 
 On the other hand, as systems become more automated, the number of sensors in the system and the amount of generated data per sensor increase. A system model for such complex systems may be prohibitively expensive \cite{doi:10.1177/0142331219879858}. Therefore, a data-driven approach might be more convenient for such complex systems \cite{https://doi.org/10.1002/qj.3551,DBLP:journals/corr/abs-1806-03753}. The hybrid estimation approach combines the advantages of both methods. It leads to fast state calculation as in the model-based approach and handles system complexity as in the data-driven approach \cite{JIN2020105962}. This paper considers the model-based approach.

The set-based estimation uses, in general, different set representations, such as intervals, polytopes, ellipsoids, zonotopes, and constrained zonotopes \cite{ALTHOFF2010233, Kurzhanskiy200726, ALAMO20051035, SCOTT2016126}. This paper considers zonotopes for set representation as they are computationally efficient and are closed under multiple set operations \cite{9838494}.

\subsection{State Estimation Under Privacy Constraints}
Existing private estimation strategies can be categorized into encryption-based methods and non-encryption-based methods. For the former, the authors in \cite{kim2019encrypted} developed a private Luenberger observer based on additively homomorphic encryption, where the state matrix only consists of integers such that the observer operates for an infinite time horizon. A partially homomorphic encryption scheme was used in \cite{proloc,zhang2020secure,emad2022privacy} to develop a state estimator. Most recent works in the latter are based on differential privacy (DP), a powerful non-cryptographic technique in quantifying individual privacy preservation. Roughly, DP refers to a property of a randomized algorithm; its output remains stable for any changes that occur to an individual in the database, therefore protecting against privacy-related attacks on the individual's information, but the accuracy of this output is negatively affected by the randomization of the DP algorithms. In \cite{6606817}, a general framework was developed for a differentially private filter. The authors in \cite{degue2017differentially} proposed a differentially private Kalman filter. The methodology has also been generalized to nonlinear systems in \cite{https://doi.org/10.1002/rnc.439}. Note that the above observers are point-wise ones, and disturbances or uncertain parameters, which are only known to be bounded, are present for some systems. Hence, set-based estimators and interval estimators can handle state estimation for such systems.
Recently, the authors in \cite{9147726} developed a differentially private interval observer with bounded input perturbation.

The main contributions of this article can be summarized as follows:
 \begin{itemize}

\item We develop an $(\epsilon,\delta)$-Approximate Differentially Private ($(\epsilon,\delta)$-ADP) set-based estimator based on the truncated additive mechanism with numerically optimized noise distribution \cite{DBLP:journals/corr/abs-2107-12957} and zonotopes for set representation. The proposed estimator preserves the privacy of sensor measurements during the estimation process with minimal utility loss. The precise positions of the estimated zonotopes are concealed. Also, the computation of the estimated zonotopes' centers and the private measurements are decoupled by nature from the computation of the generators.  
\item We evaluate the proposed $(\epsilon,\delta)$-ADP set-based estimator by comparing the deployment of a truncated optimal noise distribution with a truncated Laplace distribution in \cite{9147726}, and show that the latter leads to a more substantial noise than the former to achieve the same privacy level. Consequently, 
%the bounds of the optimal truncated noise is less than of
the proposed $(\epsilon,\delta)$-ADP set-based estimator achieves improved utility, represented as the average estimation error, over the existing method. All used data and code are publicly available \footnote {\label{footnote:DP-estimator}\url{https://github.com/mohammed-dawoud/Differentially-Private-Set-Based-Estimation-Using-Zonotopes}}.

\end{itemize}

The rest of the paper is organized as follows: The problem statement and preliminaries are presented in Section \ref{sec:sysmodel}. The algorithms are designed and evaluated in Section \ref{sec:main} and Section \ref{sec:eval}, respectively. Finally, we conclude the paper in Section \ref{sec:conc}. 

{\textbf{Notation.} Throughout this paper, we denote the vectors and scalars by lower case letters, the matrices by upper case letters, the set of real numbers by $\mathbb{R}$, the set of integers by $\mathbb{Z}$, and the set of positive integers by $\mathbb{Z}^{+}$. For a given matrix $M\in \mathbb{R}^{m\times n}$, its Frobenius norm is given by $\lVert M \rVert_{F} = \sqrt{\trace{(M^T M)}}$. We denote the $i$-th element of a vector or list $a$ by $a^{(i)}$. For a vector $a\in \mathbb{R}^n$, we denote its $L_2$ norm by $| a |_2 = \sqrt{\sum^n_{i=1}(a^{(i)})^2}$.
% \st{and $L_\infty$ norm by $\lVert y \rVert_{\infty} =\max_{i\in\{1,\dots\,n\}}{y^{(i)}}$}
 We use $\{y_k\}$ to denote a list of the measurement signals vectors $y_k$ at different time steps $k$. We denote the $L_2$ norm of a vector-valued signal $y:\mathbb{Z}^{+} \rightarrow \mathbb{R}^n$, by $\lVert y \rVert_2 = \sqrt{\sum^{\infty}_{k=1}(|y_k|_2)^2}$.

\section{Preliminaries and Problem Statement}
\label{sec:sysmodel}
In this section, we present the needed background and the problem setup.
\subsection{Set Representation and Set-based Estimation}
Now, we are going to introduce {the set representation and operations,} and the set-based estimation approach.
\subsubsection{{Set Representation and Operations}}
The zonotope is defined as follows:
\begin{definition} [\textbf{Zonotope} \cite{10.1007/978-3-540-31954-2_19}]
\label{def:zonotope} 
Given a center $c_{\mathcal{Z}} \in \mathbb{R}^n$ and $\gamma_{\mathcal{Z}} \in \mathbb{N}$ generator vectors in a generator matrix $G_{\mathcal{Z}}=\begin{bmatrix} g_{\mathcal{Z}}^{(1)}& \dots &g_{\mathcal{Z}}^{(\gamma_{\mathcal{Z}})}\end{bmatrix} \in \mathbb{R}^{n \times \gamma_{\mathcal{Z}}}$, a zonotope is defined as
%\cite[Def.~1]{Girard2005} 
\begin{align}\label{eqn:zonotope}
	\mathcal{Z} = \Big\{ x \in \mathbb{R}^n \; \Big| \; x = c_{\mathcal{Z}} + \sum_{i=1}^{\gamma_{\mathcal{Z}}} \beta^{(i)} \, g^{(i)}_{\mathcal{Z}} \, ,
	\norm{\beta^{(i)}}_\infty \leq 1  \Big\}.
\end{align}
We use the shorthand notation $\mathcal{Z} = \zono{c_{\mathcal{Z}},G_{\mathcal{Z}}}$ for zonotopes. %\hfill $\square$
\end{definition}
It is worth noting that zonotopes are closed under the linear map and Minkowski sum \cite{conf:althoffthesis}. \label{item:linear-map}
%Zonotopes are closed under linear transformation. 
The linear map $L \in \mathbb{R}^{m  \times n}$ for zonotope ${\mathcal{Z}}$ is defined and computed as follows:
\begin{align}
L \mathcal{Z} = \{Lz | z\in\mathcal{Z}\}  = \zono{L c_{\mathcal{Z}}, L G_{\mathcal{Z}} }. \label{eqn:linmap}
\end{align}

%A linear map $L$ is defined as $L \mathcal{Z}  = \zono{L c_{\mathcal{Z}}, L G_{\mathcal{Z}}}$.
%ii) \textit{Minkowski Sum}.
\label{item:minkowski-sum}
%Zonotopes are also closed under Minkowski sum.
Given two zonotopes $\mathcal{Z}_1=\langle c_{\mathcal{Z}_1},G_{\mathcal{Z}_1} \rangle$ and $\mathcal{Z}_2=\langle c_{\mathcal{Z}_2},G_{\mathcal{Z}_2} \rangle$, the Minkowski sum is defined and computed as %can be computed exactly \cite{conf:zono1998}:
\begin{align}\label{eqn:minkowski-sum}
     \mathcal{Z}_1 \oplus \mathcal{Z}_2 &= \{ z_1 + z_2 | z_1 \in \mathcal{Z}_1, z_2 \in \mathcal{Z}_2 \}\nonumber\\&= \Big\langle c_{\mathcal{Z}_1} + c_{\mathcal{Z}_2}, [G_{\mathcal{Z}_1}, G_{\mathcal{Z}_2}]\Big\rangle.
\end{align}

We define and compute the Cartesian product of two zonotopes $\mathcal{Z}_1 $ and $\mathcal{Z}_2$ by 
\begin{align}\label{eqn:cartesian-product}
\mathcal{Z}_1 \times \mathcal{Z}_2 &= \bigg\{ \begin{bmatrix}z_1 \\ z_2\end{bmatrix} \bigg| z_1 \in \mathcal{Z}_1, z_2 \in \mathcal{Z}_2 \bigg\} \nonumber\\
&= \Big\langle \begin{bmatrix} c_{\mathcal{Z}_1} \\ c_{\mathcal{Z}_2} \end{bmatrix}, \begin{bmatrix} G_{\mathcal{Z}_1} & 0 \\ 0 & G_{\mathcal{Z}_2}\end{bmatrix} \Big\rangle.
\end{align}
\subsubsection{{Set-Based Estimation}}
\label{subsubsection:set-based-estimation} 
Consider a linear discrete-time system with bounded noise given by\begin{equation}\label{eqn:system-model}
\begin{aligned}
    x_{k+1} &= F x_k + w_k,\\
    y_{k}^{(i)} &= H^{(i)} x_k + v_{k}^{(i)},
\end{aligned}
\end{equation}
where $x_k \in \mathbb{R}^n$ is the system state at time step $k \in \mathbb{Z}^+$, and $y_{k}^{(i)}\in \mathbb{R}$ is the measurement signal of single sensory entity $i$ such that $i \in \{1,\dots,\;m\}$ with $m$ equals the number of available sensors. {The matrices }$F\in \mathbb{R}^{n\times n}$ and $H^{(i)} \in \mathbb{R}^{1\times n}$ are state and measurement matrices, respectively. The vector \begin{math}w_k\end{math} and the value \begin{math} v_{k}^{(i)}\end{math} are the process and the measurement noise, respectively. They are assumed to be unknown but bounded by the zonotopes $\mathcal{Z}_w$, and  $\mathcal{Z}^{(i)}_{v}$, respectively.
The system has a bounded initial state $x_0\in \mathcal{Z}_0= \zono{c_0,G_0}$, and a predicted state set $\mathcal{\hat{S}}_{k}$. At each time step $k$, the set-based state estimator aims to find the corrected state set $\mathcal{\Bar{S}}_{k}$ by finding the intersection between the predicted state set $\mathcal{\hat{S}}_{k}$ and the measurement sets corresponding to the list of all sensory entities' measurement signals $y^{(i)}_k$, $i=1,\dots,m$ \cite{9838494,DBLP:journals/corr/abs-2010-11097,Wang2021,VALERO202011277}. The set-based state estimation is achieved through the following steps:
%\begin{itemize}
\\\textbf{The Prediction Step.} Consider an initial state $x_0 \in \mathcal{Z}_0$ and a process noise $w_k \in \mathcal{Z}_w$. The set-based state estimator determines the predicted state set $\mathcal{\hat{S}}_{k}$ which is a set of all reachable states, formulated by
\begin{align}\label{eqn:predicted-state-set}
	\mathcal{\hat{S}}_{k}  = F \mathcal{\Bar{S}}_{k-1}\oplus \mathcal{Z}_w.
\end{align}
\textbf{The Correction Step}. The corrected state set $\mathcal{\Bar{S}}_{k}$ is determined using the predicted state set $\mathcal{\hat{S}}_{k}$ \eqref{eqn:predicted-state-set}, the measurement sets corresponding to $\{y^{(i)}_k \}_{k\in\mathbb{Z^+}}$, and some weights $\lambda^{(i)}_{k}$, $i \in \{1,\dots,m\}$ according to the following lemma.
\begin{lemma}[\cite{9838494}] \label{lemma-set-based-estimation-zonotope} Let $x_0 \in \mathcal{Z}_0{=}\zono{c_0,G_0}$ and $v^{(i)}_{k} \in \mathcal{Z}^{(i)}_{v}{=}\zono{c^{(i)}_{v},G^{(i)}_{v}}$. Then, the corrected state set $\mathcal{\Bar{S}}_{k}$ is directly computed as a zonotope $\mathcal{\Bar{Z}}_k=\zono{\Bar{c}_k, \Bar{G}_k}$ over-approximating the intersection between the zonotope $\mathcal{\hat{Z}}_{k-1}=\zono{\hat{c}_{k-1}, \hat{G}_{k-1}}$ representing the predicted state set $\mathcal{\hat{S}}_{k-1}$ and the $m$ regions for $x_k$ states corresponding to $y^{(i)}_{k}$, where
\begin{align}
	\Bar{c}_k &=  \hat{c}_{k-1} + \sum_{i=1}^{m} \lambda^{(i)}_{k} \big ( y^{(i)}_{k} - H^{(i)} \hat{c}_{k-1} - c^{(i)}_{v} \big), \label{eqn:estimation-with-zonotope-center}\\
    \Bar{G}_k &=\big[(I- \sum_{i=1}^{m} \lambda^{(i)}_{k} H^{(i)})\hat{G}_{k-1}, \; \lambda^{(1)}_{k}G^{(1)}_{v}, \dots, \lambda^{(m)}_{k}G^{(m)}_{v} \big],\label{eqn:estimation-with-zonotope-generator}
\end{align}
where $\lambda^{(i)}_{k}\in \mathbb{R}^{n\times 1}$ are weights such that $ i \in \{1,\dots,\; m\}$ and $k\in \mathbb{Z}^+$.
\end{lemma}
 As in \cite{conf:disdiff}, the optimal weights $\Bar{\Lambda}_k= [ {\lambda}^{(1)}_{k},\dots,\; {\lambda}^{(m)}_{k}]$ are calculated to reduce the Frobenius norm of the generator matrix $ \Bar{G}_k$ as follows: 
\begin{align}\label{eqn:optimal-lambda-weights}
\Bar{\Lambda}_k = \argminB_{\lambda_k} \, \lVert \Bar{G}_k \rVert_{F}^{2}.
\end{align}
The prediction update aims to obtain $\mathcal{\hat{Z}}_{k}=\zono{\hat{c}_{k}, \hat{G}_{k}}$, which is computed according to the system model defined in \eqref{eqn:system-model} by
\begin{align}\label{eqn:time-update}
   \mathcal{\Acute{S}}_{k}=F \mathcal{\Bar{S}}_{k}\oplus \mathcal{Z}_w,
\end{align}
where $\mathcal{\Acute{S}}_{k}=\mathcal{\Acute{Z}}_{k}=\zono{\Acute{c}_{k}, \Acute{G}_{k}}$ and we have the reduced zonotope $\mathcal{\hat{Z}}_{k}=\zono{\hat{c}_{k}, \hat{G}_{k}}$ with $\Hat{c}_k=\Acute{c}_k$, then the order of generator matrix $\Acute{G}_k$ of $\mathcal{\Acute{Z}}_{k}$ is reduced according to \cite{10.1007/978-3-540-31954-2_19}:  
\begin{align}\label{eqn:girard-reduction}
    \hat{G}_{k} =\downarrow_q	\Acute{G}_k,
\end{align}
where $\downarrow_q$ denotes the reduction in the order of the generator matrix.
\subsection{{Differential Privacy}}

{Next, we will present a few notions in differential privacy, which will be used later to develop the differentially private set-based estimator.}
Let $\mathcal{D}$ denote a space of datasets of interest, which are the global measurement signals $y$ released as a sequence of measurement signals vectors $\{y_k\}_{k\in\{1,\dots,\mathcal{T}\}}$ where $\mathcal{T}=\infty$ is also of interest and $y_k=\big[{y_k^{(1)}},\dots,{y_k^{(m)}}\big]^T$ \cite{9147726}.  
We aim to provide a certain level of privacy protection for an individual entity's measurement signal; therefore, the state estimates should be insensitive to its contribution to the global measurement signals $y$.
Two global measurement signals $y$ and $\Acute{y}$ are called adjacent and can be denoted by Adj$(y,\Acute{y})$, if and only if they differ by the value of exactly one entity's measurement signal $y_k^{(i)}, \; i \in \{ 1,\dots,m\}$ \cite{degue2017differentially,9147726,6483414}.
Since the privacy-preserving noise is added to the measurement signals themselves,
% and not to the queries applied to them,
then the allowed variation within Adj$(y,\Acute{y})$ is bounded by what is called sensitivity, which is defined formally as follows.

\begin{definition} [\textbf{Sensitivity} \cite{9147726,6483414}]
\label{def:sensitivity}
The allowed deviation for a single sensory entity's measurement  signal between two adjacent global measurement signals $y$ and $\Acute{y}$, i.e., Adj$(y,\Acute{y})$, is bounded in the $L_2$ norm by $s$ and given by 
\begin{align}
\lVert y-\Acute{y}\rVert_2 \leq s,
\end{align}
where $s\geq0$.
\end{definition}
 It deserves noting that $y$ and $\Acute{y}$ are considered different if there is any time step $k$ at which $y_k^{(i)}\neq\Acute{y}_k^{(i)}$.
Given a pair of adjacent global measurement signals, i.e., Adj$(y,\Acute{y})$, a differentially private mechanism aims to prevent an adversary from inferring knowledge about an individual entity's measurement signal by generating randomized outputs with close distributions for adjacent inputs.
{We use a pair of non-negative constants $(\epsilon,\delta)$ to quantify the privacy loss \cite{10.1007_11787006_1,degue2017differentially}. In particular, $e^{\epsilon}$ represents the ratio of probabilities to obtain the same state estimates with adjacent global measurement signals $y$ and $\Acute{y}$, and $\delta$ allows a small degree of violation when bounding the ratio of probabilities.}

\begin{definition} [\textbf{Approximate Differential Privacy (ADP)} \cite{10.1007_11787006_1,degue2017differentially}]
\label{def:approximate-dp} 
% Let the space of measurements $y$ denoted by $\mathcal{D}$. 
A randomized mechanism ${M}$, that maps the measurement signals' space $\mathcal{D}$ equipped with a sensitivity $s$ (Definition \ref{def:sensitivity}) to an anonymized measurement signals' space $\mathcal{O}$, is $(\epsilon,\delta)$-ADP if $\forall$ $\mathcal{H} \in \mathcal{O}$ and $\forall$ Adj$(y,\Acute{y}) \in \mathcal{D}$
\begin{align}\label{eqn:adp-formula}
	Pr\bigl[M(y)\in\mathcal{H}\bigl]\leq e^{\epsilon} Pr\bigl[M(\Acute{y})\in\mathcal{H}\bigl]+\delta, 
\end{align}
such that $\epsilon,\delta \geq 0$. If $\delta=0$, ${M}$ is said to be $\epsilon$-differentially private.
%\hfill $\square$
\end{definition}

\subsection{{Problem} Setup}
Consider the following cloud-based state estimation setup \cite{DBLP:journals/corr/abs-2010-11097}:
\begin{itemize}
\item \textbf{Plant.} A system that we aim to estimate its set of possible states, described by the publicly known linear discrete-time model in \eqref{eqn:system-model} with unknown but bounded disturbances.
\item \textbf{Sensors.} An array of $m$ sensory entities that provides a vector of measurement signals $y_k=\big[{y_k^{(1)}},\dots,{y_k^{(m)}}\big]^T$, $k\in \mathbb{Z^+}$. Each sensory entity $i$ produces a private measurement signal $y^{(i)}_k$, $i\in \{1,\dots, m\}$.
\item \textbf{Sensor Manager.} A trusted entity with computational capabilities that enable it to aggregate the measurement signals of all sensory entities and perturb them with the differential privacy noise. 
\item \textbf{Cloud Estimator.} An untrusted  entity that performs set-based estimation for the system state.
\end{itemize}

\textbf{Problem.} {We aim to design a differentially private set-based estimator for the plant with the model in \eqref{eqn:system-model} such that the state estimates are obtained while keeping each sensory entity's measurement signal $(\epsilon,\delta)$-ADP and protecting it from any untrusted entity (e.g., the cloud estimator)}.

\section{Differentially Private Set-Based Estimation}\label{sec:main}

In this section, we develop a differentially private version of the set-based estimator.

In subsection \ref{subsection:Design-of-Truncated-Optimal-Additive}, we use an additive noise mechanism (Definition \ref{def:additive-noise-mechanism}) to perturb the measurement signals by an optimal noise with a numerically generated truncated distribution (Definition \ref{def:optimal-noise}), satisfying Lemma \ref{lemma:dp-optimal-noise-mechanism}, and obtain $(\epsilon,\delta)$-ADP measurement signals. Then, in subsection \ref{subsection:Differentially-Private-Zonotope-based-Set-Membership-Estimation}, Theorem \ref{theorem:dp-estimator}, deploys these $(\epsilon,\delta)$-ADP measurement signals to the set-based estimator to obtain state estimates that are insensitive to a single sensor's measurement signal. The full $(\epsilon,\delta)$-ADP estimator scenario is summarized by Algorithm \ref{alg:adp-set-estimation-using-zonotope}.

\subsection{Design of Truncated Optimal Additive Noise}\label{subsection:Design-of-Truncated-Optimal-Additive}

Additive noise mechanisms such as the Gaussian and Laplace mechanisms typically perturb the measurement signals with particular random noise to achieve DP. 
However, most of the results require the support of the noise distribution to be unbounded, except for a few cases such as \cite{Liu_2019,Croft2022}. This class of mechanisms is not 
% Due to the nature of set-based estimation and safety considerations, they \tr{unclear reference} are not directly 
applicable for some real-time applications, such as safety-critical systems \cite{DBLP:journals/corr/abs-2107-12957}. One closely related work \cite{9147726} developed an interval observer with DP based on a truncated Laplace mechanism. However, the design relies on an $L_1$ norm-based adjacency relation and an analytical bound for the noise variance, which may cause the design to be conservative, consuming more noise to achieve DP (a numerical comparison is provided in Section \ref{sec:conc}).
To solve this issue, we follow the numerical approach, recently developed in \cite{DBLP:journals/corr/abs-2107-12957}, to optimize a noise distribution that is subjected to a bounded support constraint and the privacy constraint in Definition \ref{def:approximate-dp}.

In our setup, all sets are defined over the continuous domain of all real numbers. However, the truncated optimal noise distribution that we employ to achieve DP is a distribution that is generated numerically and consists of discrete noise occurrence probabilities with noise values selected from the continuous domain. Hence, the resultant domain of adding that noise to the sets is still continuous. 

Next, we define a class of truncated noise distribution functions in \eqref{eqn:optimal-noise-function}. Then, for a fixed DP parameter $\epsilon$ and a particular noise model given in \eqref{eqn:noise-model}, we present an optimization problem, where the objective function 
% is the loss function in \eqref{eqn:optimization-function}, which 
balances between the privacy loss parameter $\delta$ and the utility loss. 
Upon solving this optimization problem, an optimal noise distribution is generated. 

\begin{definition}[\textbf{Truncated Noise Distribution} \cite{DBLP:journals/corr/abs-2107-12957}]\label{def:optimal-noise}
Let ${\Phi}=[-d,d]$ define a bounded noise range such that $d \in \mathbb{R}$, and $\Acute{\Phi}=\{\phi_l\}_{l\in \{1,\dots,\;2N\}}$ be its discretization on $2N$ equidistant steps such that $\phi_l\in[-d,d]$, $N\in \mathbb{Z}^+$, and $l\in \{1,\dots,\;2N\}$,
then a numerically generated, truncated, discrete, equidistant, symmetric, and monotonically decreasing from its zero center noise distribution function, denoted by $P(\phi_l)$, has the following properties
% is defined as 
\begin{subequations}\label{eqn:optimal-noise-function}
\begin{equation}
    \sum_{\phi_l\in \Acute{\Phi}} P(\phi_l)=1\;\; \mathrm{and}\;\; P(\phi_l)\geq 0. \;\;\;\; \;\;(\textrm{Distribution})
\end{equation}    
    \begin{equation}
    P(\phi_l) \geq P(\phi_m)\;\;\; \forall\; \phi_m>\phi_l>0, \;\;\;\;(\textrm{Monotonicity})    
    \end{equation}
\quad \qquad where $l,m\in \{N+1,\dots,\;2N\}$.
    \begin{equation}
    P(\phi_l) = P(-\phi_l).\;\;\;\;\;\;\; \;\;\;\;\;\;\;\;\;\; (\textrm{Symmetry})      
    \end{equation}
    \end{subequations}
% while satisfying the $(\epsilon, \delta)$-ADP constraints.
\end{definition}

Note that the case with arbitrarily dimensional and spherically rotation-symmetric noise distributions and sensitivity conditions can be reduced to a 1-dimensional privacy analysis \cite{10.1145/2976749.2978318}. Based on this claim, we can define the additive noise for the case in which the global measurement signal $y$ contains an array of $m$ sensors.

\begin{definition}[\textbf{Additive Noise Mechanism}
% \cite{DBLP:journals/corr/abs-2107-12957}
]\label{def:additive-noise-mechanism}
Given a vector of measurement signals $y_k$ and a noise vector $\phi_k \in \Acute{\Phi}$ of independent and identically distributed (IID) coordinates with the probability distribution in Definition \ref{def:optimal-noise} and satisfying Lemma \ref{lemma:dp-optimal-noise-mechanism}, and the successive samples of $\phi_k$ are also IID, then the additive noise mechanism $M_y$ is defined as
\begin{equation}\label{eqn:additive-noise-mechanism}
    {M}_y:\hat{y}^{DP}_k = y_k + \phi_k, 
\end{equation}
where %$\hat{y}_k$ is a vector of the anonymized measurement signals.
$\hat{y}^{DP}_k$ is a vector of $(\epsilon, \delta)$-ADP measurement signals.
\end{definition}

\begin{lemma}[\textbf{\cite [Theorem 15]{DBLP:journals/corr/abs-2107-12957}}]\label{lemma:dp-optimal-noise-mechanism}
Let $M_{y}$ be an additive noise mechanism
%(Definition \ref{def:additive-noise-mechanism})
with a sensitivity $s$ (Definition \ref{def:sensitivity}), a vector of measurement signals $y_k$
, and $\Acute{\Phi}=\{\phi_l\}_{l\in \{1,\dots,\;2N\}}$ be the discretization of ${\Phi}$ with the truncated optimal noise distribution $P(\phi_l)$ (Definition \ref{def:optimal-noise}).
% Let $d \in \{\phi_l - \phi_m|l,m \in \mathbb{Z}\}, d\leq s$
If $\forall\;\hat{\Phi}\subseteq \Acute{\Phi}$
\begin{equation}\label{eqn:dp-optimal-noise}
    \sum_{\phi_l\in \hat{\Phi}}P(\phi_l) \leq e^{\epsilon}{\sum_{\phi_l\in \hat{\Phi}}P(\phi_l+s)}+\delta, \end{equation}
then
%for any vector of measurement signals $y_k\in \mathbb{D}$, 
the additive noise mechanism $M_{y}$ is $(\epsilon, \delta)$-ADP for any $y_k\in \mathcal{D}$.
%, and the vector of the anonymized measurement signals, denoted by $\hat{y}^{DP}_k$, also becomes a vector of $(\epsilon, \delta)$-ADP measurement signals.
\end{lemma}

It deserves noting that, given that the sensitivity $s$ (Definition \ref{def:sensitivity}) is satisfied $\forall\;k\in \mathbb{Z^+}$, the mechanism $M_y$ is $(\epsilon, \delta)$-ADP at any time step $k$ \cite[Lemma 2]{6606817}. 

    Next, we present the loss function, 
    denoted by $L^{\Omega_t}_{\gamma}$. % \st{, for learning the optimal noise distribution $P(\phi_l)$}
    This function balances between the privacy parameter $\delta$ and a utility loss $U$ at a fixed $\epsilon$ and is given by \begin{subequations} \label{eqn:optimization-function}
\begin{equation} \label{eqn:optimization-function-1}    L^{\Omega_t}_{\gamma} = \delta+\Omega_t U,
\end{equation}
where the utility loss $U$ is given by
\begin{equation}
  U=\Big(|\phi_l|^\gamma\sum_{\phi_l \in \Acute{\Phi}}P(\phi_l)\Big)^{1/\gamma}   
\end{equation}
with $\gamma\in\{1,2\}$ such that $\gamma$ selects between $L_1$ or $L_2$ norm-based utility loss, and
\begin{equation}
    \Omega_t= \max\bigg(\frac{\Omega_{start}}{2^{t/\Gamma}},\Omega_{min}\bigg)
\end{equation}
\end{subequations}
is the utility weight at training epoch $t$ where $t\in \mathbb{Z^+}$ with an exponentially decaying rate $\Gamma$ from a starting value $\Omega_{start}$ and with a lower bound $\Omega_{min}$. The optimal noise distribution $P(\phi_l)$ is then optimized by minimizing the weighted sum of the utility loss $U$ and the privacy parameter $\delta$. This noise distribution is generated through the following steps, and we are going to omit $(\phi_l)$ to ease the notation. 
We start by generating the first monotonically increasing half (i.e., $\{\phi_l\}_{l\in \{1,\dots,\; N\}}$) of the noise distribution $P(\phi_l)$, which is given by
\begin{subequations} \label{eqn:noise-model}
\begin{equation}\label{eqn:noise-model-2}
    P_l=1/2 \;\textrm{SoftMax}(r_l);\;r_l \in \{r_0,\dots,\;r_N\},
\end{equation}
where SoftMax($r_l$) = $e^{r_l}/\sum^{N}_{i=0}e^{r_i}$ and $N$ is the number of discretization steps in the half-width of the noise distribution $P(\phi_l)$. The SoftMax function normalizes the $r_l$ values into a distribution, and the $r_l$ values are generated using a model of $v$-stacked Sigmoid functions (i.e., $\sigma(\phi_l)=(1+e^{-\phi_l})^{-1}$), which is given by
\begin{equation}\label{eqn:noise-model-1}
    r_l= \ln \bigg[A^2+\sum_{j=0}^v B_{j}^2\cdot \sigma(C\cdot(\phi_l-F_j)) \bigg],
\end{equation}
where $\phi_l=l\frac{d}{N}-d$. The parameters $A,\; B_j,\;C$, and $F_j$ are randomly initialized, then learned to optimize the loss function in \eqref{eqn:optimization-function} using numerical optimization methods (e.g., stochastic gradient descent (SGD)).

Next, the first half (i.e., $\{\phi_l\}_{l\in \{1,\dots,\; N\}}$) of the noise distribution generated by \eqref{eqn:noise-model-2} and denoted by $P_l$ is mirrored according to the following:
\begin{equation}
    P_j=P_{2N-j+1}\;\; \mathrm{for}\; j\in\{N+1,\cdots,\;2N\}. 
\end{equation}
Then, $P_j$ (i.e., $\{\phi_l\}_{l\in \{N+1,\dots,\;2N\}}$) is concatenated to $P_l$ (i.e., $\{\phi_l\}_{l\in \{1,\dots,\;N\}}$) to obtain the symmetric noise distribution $P(\phi_l)$ (i.e., $\{\phi_l\}_{l\in \{1,\dots,\;2N\}}$). 
% as shown in Figure \ref{fig:optimal-noise-distribution}.
\end{subequations}

\subsection{Differentially Private Zonotope-based Set-Membership Estimation}\label{subsection:Differentially-Private-Zonotope-based-Set-Membership-Estimation}
Now, we have a vector of $(\epsilon,\delta)$-ADP measurement signals $\hat{y}^{DP}_k$ determined by the sensor manager using a vector of the sensory entities' measurement signals $y_k$, the additive noise mechanism $M_y$ (Definition \ref{def:additive-noise-mechanism}), and the truncated optimal noise distribution $P(\phi_l)$ satisfying \eqref{eqn:optimal-noise-function}, \eqref{eqn:optimization-function}, \eqref{eqn:noise-model}, and Lemma \ref{lemma:dp-optimal-noise-mechanism}. Next, we will deploy these $(\epsilon,\delta)$-ADP measurement signals to the set-based estimator (Lemma \ref{lemma-set-based-estimation-zonotope}), i.e., the cloud estimator, to preserve the privacy of sensor measurements during the estimation process and get an $(\epsilon,\delta)$-ADP differentially private version of that set-based estimator in the following theorem.

\begin{theorem}\label{theorem:dp-estimator}
Given an $(\epsilon,\delta)$-ADP additive noise mechanism $M_y$ (Definition \ref{def:additive-noise-mechanism}) with sensitivity $s$ (Definition \ref{def:sensitivity}) and the truncated optimal noise distribution $P(\phi_l)$ satisfying \eqref{eqn:optimal-noise-function}, \eqref{eqn:optimization-function}, \eqref{eqn:noise-model}, and Lemma \ref{lemma:dp-optimal-noise-mechanism}, then using a vector of $(\epsilon,\delta)$-ADP measurement signals $\hat{y}^{DP}_k$, an approximate differentially private $(\epsilon,\delta)$-ADP version of the set-based estimator (Lemma \ref{lemma-set-based-estimation-zonotope}) which obtains state estimates that are insensitive to a single sensor’s measurement signal can be defined as  
\begin{align}
	\Bar{c}^{DP}_k = & \hat{c}^{DP}_{k-1} + \sum_{i=1}^{m+1} \lambda^{(i)}_{k} \big ( \hat{y}^{(i)DP}_{k} - H^{(i)} \hat{c}^{DP}_{k-1} - c^{(i)}_{v}-c_{p} \big),\label{eqn:dp-estimation-with-zonotope-center} \\
   \nonumber \Bar{G}^{DP}_k{=} &\Big[(I- \sum_{i=1}^{m+1} \lambda^{(i)}_{k} H^{(i)})\hat{G}^{DP}_{k-1},\\& \lambda^{(1)}_{k}G^{(1)}_{v},\dots, \lambda^{(m)}_{k}G^{(m)}_{v},\lambda^{(m+1)}_{k}G_{p}  \Big],\label{eqn:dp-estimation-with-zonotope-generator}
\end{align}
where $\mathcal{\Bar{S}}^{DP}_k=\mathcal{\Bar{Z}}^{DP}_k=\zono{\Bar{c}^{DP}_k, \Bar{G}^{DP}_k}$ is the corrected state set, $\mathcal{\hat{S}}^{DP}_{k-1}=\mathcal{\hat{Z}}^{DP}_{k-1}=\zono{\hat{c}^{DP}_{k-1}, \hat{G}^{DP}_{k-1}}$ is the predicted state set, ${M}_y:\hat{y}^{DP}_k = y_k + \phi_k$, $c_{p}$ is the noise value at the center of the truncated optimal noise distribution, and $G_{p}$ is the generator matrix, which is created using the range of the optimal noise $d$.  
\end{theorem}
\begin{proof}
We have $\hat{y}^{DP}_k$ is $(\epsilon,\delta)$-approximate differentially private according to Lemma \ref{lemma:dp-optimal-noise-mechanism}, then based on the post-processing property \cite[Proposition~2.1]{TCS-042}, the privacy guarantees of approximate differential privacy still hold after any arbitrary processing. Hence, the set of states estimated using $\hat{y}^{DP}_k$ and represented by $\mathcal{\Bar{Z}}^{DP}_k$ is insensitive to a single sensor's measurement signal.
Additionally, we use bounded noise to preserve the privacy of each sensory entity’s measurement signal, so estimation errors are also bounded. 
\end{proof}
It is worth mentioning that the generators of all zonotopes are not corrupted by the $(\epsilon,\delta)$-ADP optimal noise, which in turn preserves the same state containment guarantees \cite{DBLP:journals/corr/abs-2010-11097}.
Likewise in \eqref{eqn:estimation-with-zonotope-generator}, $\lambda^{(i)}_{k}\in \mathbb{R}^{n\times 1}$ are weights such that $ i \in \{1,\dots,\; m+1\}$. Again, as in \cite{conf:disdiff}, the optimal weights $\lambda^{(i)}_{k}$ are calculated by 
\begin{equation}\label{eqn:dp-optimal-lambda-weights}
\Bar{\Lambda}_k^* = \argminB_{{\lambda}_k}\lVert \Bar{G}^{DP}_k \rVert_{F}^{2},
\end{equation}
where $\Bar{\Lambda}_k^* = [ {\lambda}^{(1)}_{k},\dots,\; {\lambda}^{(m+1)}_{k}]$. The optimal weights $\Bar{\Lambda}_k^*$ are calculated to reduce the Frobenius norm of the generator matrix $\Bar{G}^{DP}_k$; therefore, they reduce uncertainty around estimated values. The prediction update according to the system model defined in \eqref{eqn:system-model}, is given by
\begin{subequations}
\begin{equation}\label{eqn:dp-time-update}
    \mathcal{\Acute{S}}^{DP}_{k}=F \mathcal{\Bar{S}}^{DP}_{k}\oplus \mathcal{Z}_w,
\end{equation}
where $\mathcal{\Acute{S}}^{DP}_k$ is represented by the zonotope $\mathcal{\Acute{Z}}^{DP}_k=\zono{\Acute{c}^{DP}_k, \Acute{G}^{DP}_k}$.

Next, the order of generator matrix $\Acute{G}^{DP}_k$ is reduced according to \cite{10.1007/978-3-540-31954-2_19}, 
then we get $\mathcal{\Hat{S}}^{DP}_k=\mathcal{\Hat{Z}}^{DP}_k=\zono{\Hat{c}^{DP}_k,\Hat{G}^{DP}_k}$.
\end{subequations}

Finally, we summarize the ($\epsilon, \delta$)-ADP set-based estimator in Algorithm \ref{alg:adp-set-estimation-using-zonotope}. 
\begin{algorithm}
\caption{($\epsilon, \delta$)-ADP Set-Estimation Using Zonotope}\label{alg:adp-set-estimation-using-zonotope}
\begin{algorithmic}[1]
{\Statex {\bfseries Input:} DP parameter $\epsilon$, noise range $d$, sensitivity $s$.
} 
{ \Statex {\bfseries Output:} $\mathcal{\hat{S}}^{DP}_{k}=\mathcal{\Hat{Z}}^{DP}_k=\zono{\hat{c}^{DP}_k, \hat{G}^{DP}_k}$.}
\Statex{\bfseries Initialization:} The sensor manager uses $\epsilon$, $d$, and $s$ to generate the truncated optimal noise distribution $P(\phi_l)$ according to \eqref{eqn:optimal-noise-function} while optimizing the loss function in \eqref{eqn:optimization-function} by learning the parameters of the noise model in \eqref{eqn:noise-model}. 
Set $k = 1$ and $\mathcal{\hat{S}}_{0} = \mathcal{Z}_0$.
\While{$True$}
\State An array of $m$ sensors where each sensor provides a private measurement signal $y_k^{(i)}$, $ i\in \{1,\dots,\; m\}$.

\State A sensor manager aggregates the $m$ measurement signals into a vector of measurement signals; $y_k=\big[{y_k^{(1)}},\dots,{y_k^{(m)}}\big]^T$. 

{\State The sensor manager uses the additive noise mechanism $M_y$ in Definition \ref{def:additive-noise-mechanism} to obtain 
$\hat{y}^{DP}_k$; ${M}_y$: $\hat{y}^{DP}_k = y_k + \phi_k$. 
% \;\;${M}_{y}: \hat{y}^{DP}_k= y_k + \phi_l$
}
\State The cloud estimator uses $\hat{y}^{DP}_k$ to perform the estimation process according to the following steps:

\State Compute the optimal weights as in \cite{conf:disdiff}: $\Bar{\Lambda}_k^* = \argminB_{{\lambda}_k}\lVert \Bar{G}^{DP}_k \rVert_{F}^{2}$.
\State Compute the corrected state set: $\mathcal{\Bar{S}}^{DP}_{k}=\mathcal{\Bar{Z}}^{DP}_k= \zono{\Bar{c}^{DP}_k, \Bar{G}^{DP}_k}$,
\State $\Bar{c}^{DP}_k =  \hat{c}^{DP}_{k-1} + \sum_{i=1}^{m+1} \lambda^{(i)}_{k} \big ( \hat{y}^{(i)DP}_{k} - H^{(i)} \hat{c}^{DP}_{k-1} - c^{(i)}_{v}-c_{p} \big)$, 
\State $\Bar{G}^{DP}_{k} =\bigg[(I- \sum_{i=1}^{m+1} \lambda^{(i)}_{k} H^{(i)})\hat{G}^{DP}_{k-1},$\par
$\quad \quad \quad \; \; \lambda^{(1)}_{k} G^{(1)}_{v},\cdots, \lambda^{(m)}_{k} G^{(m)}_{v}, \lambda^{(m+1)}_{k}{G}_{p}\bigg]$.
 
\State Perform prediction update: $ \mathcal{\Acute{S}}^{DP}_{k}=F \mathcal{\Bar{S}}^{DP}_{k}\oplus \mathcal{Z}_w$. 
\State Reduce the order of $\Acute{G}^{DP}_k$ as in \cite{10.1007/978-3-540-31954-2_19}: $\hat{G}^{DP}_{k} =\downarrow_q	\Acute{G}^{DP}_k$. 
\State $\Hat{c}^{DP}_k = \Acute{c}^{DP}_k$.
\State Compute the predicted state set: $\mathcal{\hat{S}}^{DP}_{k}=\mathcal{\hat{Z}}^{DP}_k= \zono{\hat{c}^{DP}_k, \hat{G}^{DP}_k}$. 
\State Update the time: $k=k+1$.
\EndWhile
\end{algorithmic}
\end{algorithm}

\section{Evaluation} \label{sec:eval}

In this section, we present the experimental results for the Matlab 2018 implementation of our proposed $(\epsilon,\delta)$-ADP set-based estimator. For zonotope operations, we used the CORA toolbox.

Consider a network of $8$ sensory nodes tracking the location of an object in a circular motion within two-dimensional space with dimensions of 180m $\times$ 180m \cite{conf:disdiff}. We aim to keep the sensory nodes' measurements private from the cloud estimator while estimating the location of the rotating object. 
The state of each node consists of two variables that represent the position of the rotating object. The state and measurement matrices are 
\begin{equation*}
 F=\begin{bmatrix}
0.9920 & -0.1247 \\
0.1247 & 0.9920 
\end{bmatrix},\;
H^{(i)}=\begin{bmatrix}
1 &0 \\ 
\end{bmatrix}\; \mathrm{or}\;
\begin{bmatrix}
0 &1 
\end{bmatrix},
\end{equation*}
where the measurement matrix $H^{(i)}$ is selected based on the measurement sequence. The measurement and process noise zonotopes $\mathcal{Z}^{(i)}_{v}$ and $\mathcal{Z}_{w}$ are set to 
\begin{equation*}
    z^{(i)}_{v}=\Big\langle{\begin{bmatrix} 0 \end{bmatrix},\begin{bmatrix} 0.01\;\; 0.02 \end{bmatrix}} \Big\rangle
,\;
    z_{w}=\Big\langle{\begin{bmatrix} 0 \\0 \end{bmatrix},\begin{bmatrix} 0.50 & 0 \\0 & 0.50 \end{bmatrix}} \Big\rangle.
\end{equation*}
The $(\epsilon,\delta)$-ADP truncated optimal noise distribution (Definition \ref{def:optimal-noise}) is generated according to the noise model \eqref{eqn:noise-model} with $C=500$ in order to allow sudden jumps for the noise function $P(\phi_l)$. We set $d\in[-7,\;7]$, $\epsilon=0.3$, and $s=1$ and the parameters $A$, $B_j$, $F_j$, and $\delta$ are learned to optimize the loss function \eqref{eqn:optimization-function} using the SGD tool \cite{DBLP:journals/corr/abs-2107-12957}.

Figure \ref{fig:Object-tracking} represents a random snapshot for tracking the rotating object. We notice that the rotating object is enclosed by the estimated zonotope, which indicates that the state containment is still guaranteed, and the center of the estimated zonotope is very close to the rotating object, which is a good utility indicator. Supporting the above, Figure \ref{fig:estimated-location-error-with-epsilon} illustrates that the utility loss represented by the standard deviations of errors in the estimated locations of the rotating object increases as the privacy budget $\epsilon$ decreases.
It also shows that utility losses and average errors in estimated locations of the rotating object increase with a wider $(\epsilon,\delta)$-ADP optimal noise range. Hence, an industrial application should select the appropriate noise range based on the accepted error ranges.

For comparison, we consider the work in \cite{9147726}, where a differentially private interval estimator deploying truncated Laplace noise is presented. 
In particular, we compare the utility loss of $(\epsilon,\delta)$-ADP truncated optimal noise (Definition \ref{def:optimal-noise}) and truncated Laplace noise \cite{9147726} at $\epsilon=0.3$. The results are shown in Figure \ref{fig:Estimation-error-vs-noise-acc-vs-optimal}. The truncated Laplace noise range, for a given $\epsilon,\delta$, and sensitivity $s$, is determined by 
\begin{equation}
a=\frac{s}{\epsilon}\ln{\Big(1+e^{\epsilon}\frac{1-e^{-\epsilon}}{2\delta}\Big)}.    
\end{equation}
At a certain $\delta$ value, we find that the truncated Laplace noise causes a higher utility loss than the $(\epsilon,\delta)$-ADP truncated optimal noise. Similarly, Figure \ref{fig:deltavs-noise-range} indicates that the truncated Laplace noise needed to achieve a certain $\delta$ is wider than the $(\epsilon,\delta)$-ADP truncated optimal noise needed to achieve the same $\delta$. Indeed, for a certain privacy budget $\epsilon$, learning the truncated optimal noise distribution $P(\phi_l)$ (Definition \ref{def:optimal-noise}) using the SGD tool in \cite{DBLP:journals/corr/abs-2107-12957} allows our proposed $(\epsilon,\delta)$-ADP set-based estimator (Algorithm \ref{alg:adp-set-estimation-using-zonotope}) to minimize loss of privacy and utility simultaneously. Finally, our proposed $(\epsilon,\delta)$-ADP set-based estimator (Algorithm \ref{alg:adp-set-estimation-using-zonotope}) uses zonotopes for set representation, and the zonotopes are less conservative and allow efficient computation of the linear maps and Minkowski sum, which are essential operations at the set-based estimator.

\begin{figure}[h]
\graphicspath{ {./Figures/} }
    \centerline{
\includegraphics
[scale=0.30]{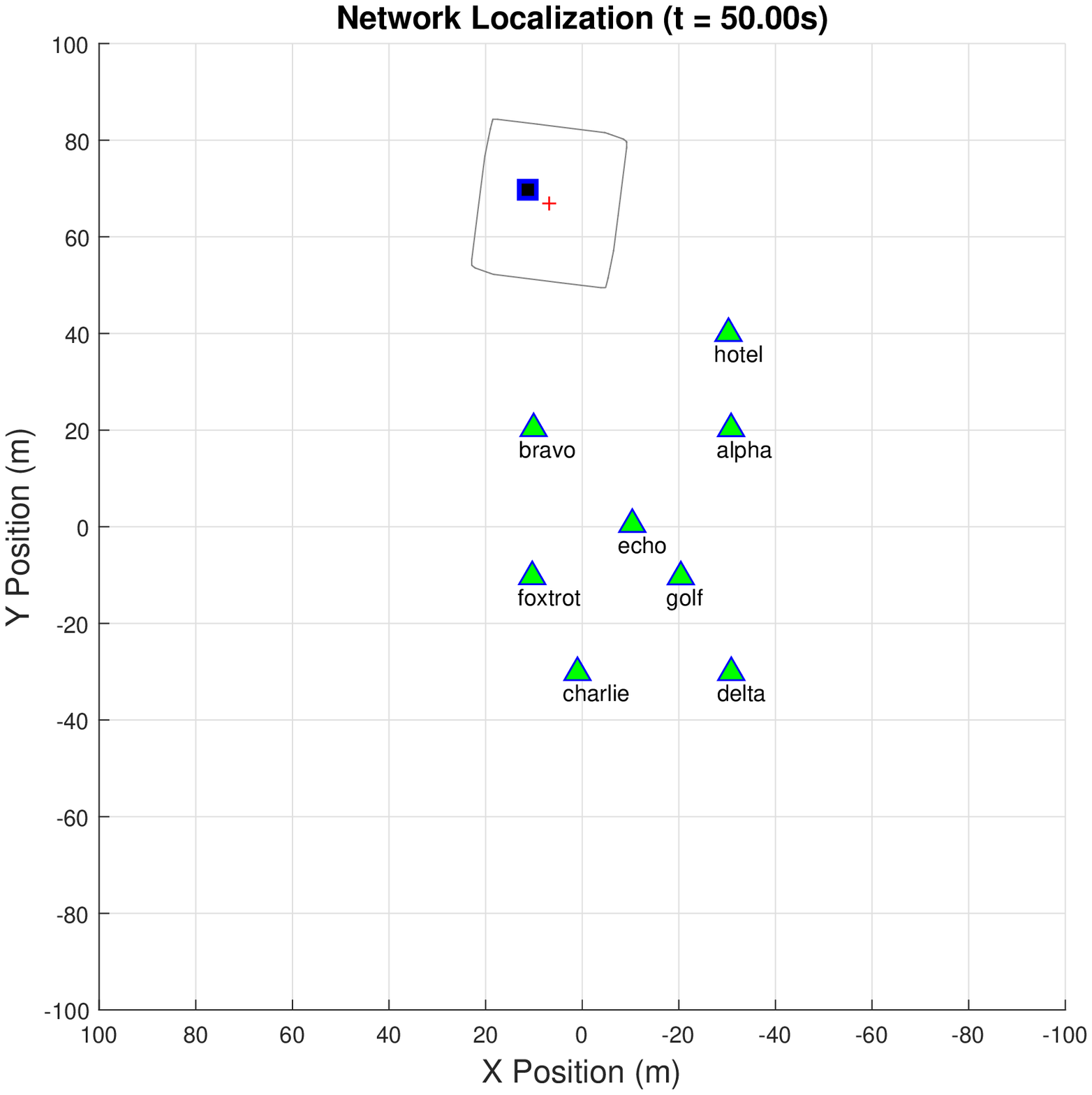}}
    \caption{Tracking an object in circular motion,  {\img{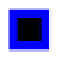}}: rotating object, {\img{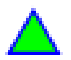}}
    : sensory nodes, {\color{red}+} : center of the estimated zonotope.}
    \label{fig:Object-tracking}
\end{figure}

\begin{table*}[htbp]\centering 
%\ra{1.3}
\begin{small}
\begin{tabular}{@{}lrrrrrrrrrrrr@{}}\toprule
{$\epsilon$\;/\;Noise Range $d(m)$} & {3} & {5} & {7} & {9} & {11} & {13} & {15} \\ \midrule
{0.1} & 0.1502\ & 0.0811\ & 0.0518\ & 0.0360\ & 0.0262\ & 0.0197\ & 0.0151\ \\ 
{0.3} & 0.1198\ & 0.0503\ & 0.0244\ & 0.0126\ & 0.0067\ &0.0036\ & 0.0020\ \\ 
{0.5} & 0.0931\ & 0.0290\ & 0.0101\ & 0.0036\ & 0.0013\ & 0.0005\ & 0.0002\ \\ %\hdashline
{0.7} & 0.0707\ & 0.0158\ & 0.0038\ & 0.0009\ & 0.0002\ & $5.64e^{-5}$\ & $1.39e^{-5}$\ \\ %\hdashline
\bottomrule
\end{tabular}
\end{small}
\caption{Optimal $\delta$ values corresponding to different ranges $d(m)$ of $(\epsilon,\delta)$-ADP optimal noise at $\epsilon={0.1,0.3,0.5,0.7}$.}
\end{table*}

\section{Conclusion} \label{sec:conc}
We have proposed an $(\epsilon,\delta)$-ADP set-based estimator that performs set-based estimation in a privacy-preserving manner so that an adversary cannot learn the actual values of measurement signals from state estimates. It deploys an additive mechanism with truncated optimal noise that holds the measurement signals private while minimizing the utility loss. The evaluation results for our proposed estimator illustrate that a wider range of the truncated Laplace noise than of the $(\epsilon,\delta)$-ADP truncated optimal noise is needed to achieve a specific value of $\delta$. Hence, the truncated optimal noise has a lower utility loss than the truncated Laplace noise with the same $\epsilon$ and $\delta$ values. Also, using zonotopes for set representation in our proposed $(\epsilon,\delta)$-ADP set-based estimator gives it a computational advantage.

\begin{figure}[t]
\graphicspath{ {./Figures/} }
    \centerline{
\includegraphics
[scale=0.6]{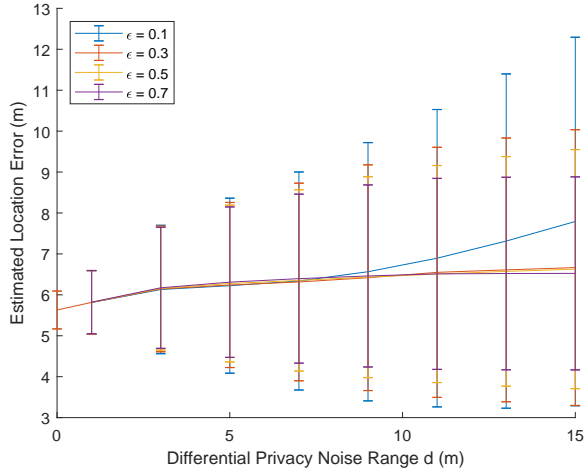}}
    \caption{Effect of $(\epsilon,\delta)$-ADP optimal noise range on the average error in the estimated location of the rotating object. }
    \label{fig:estimated-location-error-with-epsilon}
\end{figure}
\begin{figure}[t]
\graphicspath{ {./Figures/} }
\centerline{\includegraphics[scale=0.6]{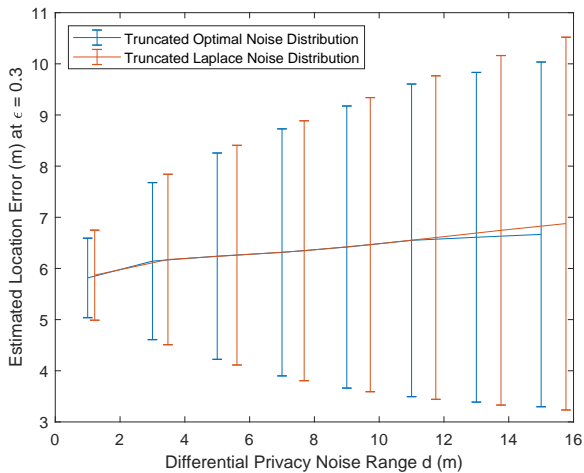}}
\caption{Comparing the effect of $(\epsilon,\delta)$-ADP truncated optimal noise and truncated Laplace noise on the error at the estimated location of the rotating object at $\epsilon=0.3$.}
\label{fig:Estimation-error-vs-noise-acc-vs-optimal}
\end{figure}
\begin{figure} [t]
\graphicspath{ {./Figures/} }
\centerline{\includegraphics[scale=0.6]{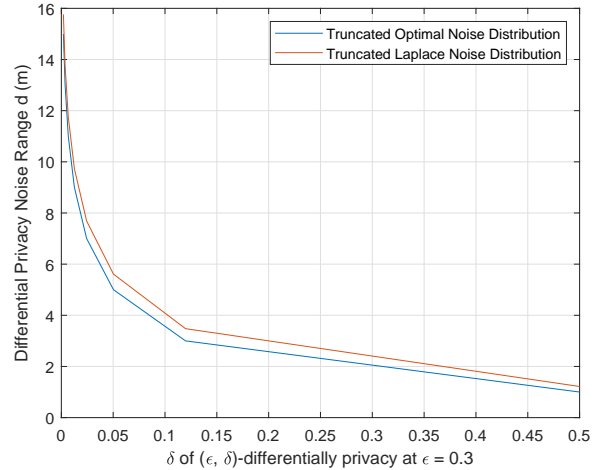}}
\caption{Comparing $\delta$ values with needed range of both $(\epsilon,\delta)$-ADP truncated optimal noise and truncated Laplace noise. }
\label{fig:deltavs-noise-range}
\end{figure}

\small
\bibliographystyle{ieeetr}
\bibliography{references}

\end{document}